\documentclass[12pt]{extarticle}
\usepackage{amsmath, amsthm, amssymb, color}
\usepackage{hyperref}
\usepackage{xcolor}
\hypersetup{
     colorlinks = true,
     linkcolor = blue,
     anchorcolor = blue,
     citecolor = red,
     filecolor = blue,
     urlcolor = blue
     }
\usepackage[sort]{cite}
\usepackage{graphicx}
\usepackage{caption}
\usepackage{mathtools}
\usepackage{enumerate}
 \usepackage{verbatim}
\usepackage{tikz,tikz-cd,tikz-3dplot}
\usepackage{amssymb}
\usetikzlibrary{matrix}
\usetikzlibrary{arrows}
\usepackage[ruled,linesnumbered]{algorithm2e}
\usepackage{caption}
\usepackage[normalem]{ulem}
\usepackage{subcaption}
\usepackage{multicol}
\usepackage{listings}
\definecolor{jgreen}{RGB}{0,102,51}
\definecolor{maroon}{RGB}{133, 5, 63}
\lstdefinelanguage{code}{
basicstyle=\small\ttfamily,
alsoletter=",
classoffset=1,
morekeywords={@var, regions, critical_points, membership},
keywordstyle={\color{maroon}},
classoffset=2,
morekeywords={using},
keywordstyle={\color{jgreen}},
xleftmargin=1.5cm,
xrightmargin=1em,
columns=fullflexible,
keepspaces=true,
stepnumber=1,
numbers=none,
captionpos=b,
showspaces=false,
frame=none
}

\oddsidemargin \evensidemargin
\usepackage{makecell}
\usepackage{array}
\usepackage{enumitem}

\newtheorem{theorem}{Theorem}

\newtheorem{proposition}[theorem]{Proposition}

\newtheorem{corollary}[theorem]{Corollary}

\theoremstyle{definition}

\newtheorem{example}[theorem]{Example}

\newcommand{\PP}{\mathbb{P}}
\newcommand{\RR}{\mathbb{R}}

\newcommand{\CC}{\mathbb{C}}

\title{\bf Computing Arrangements of Hypersurfaces}
\author{Paul Breiding, Bernd Sturmfels and Kexin Wang} \date{}

\begin{document} \maketitle

\begin{abstract}  \noindent
  We present a Julia package {\tt HypersurfaceRegions.jl}
  for computing all connected components in the complement of an arrangement of 
  real   algebraic   hypersurfaces in~$\RR^n$.
  \end{abstract}

\section{Introduction}

Arrangements of real hyperplanes  are ubiquitous in
combinatorics, algebra and geometry. 
Their {\em regions} (i.e.~connected components) are convex polyhedra, either bounded or unbounded, and their numbers are invariants of the underlying matroid~\cite{Sta}.
The number of bounded regions agrees with the
Euler characteristic of the complex arrangement complement \cite{CHKS, CDFV},
and also with the maximum likelihood degree (ML degree). The latter is the number of complex
critical points of the associated models in statistics and physics \cite{ST}.

Arrangements of nonlinear hypersurfaces are equally important, 
but they are studied much less. Polynomials of higher degree
create features that are not seen when dealing with hyperplanes.
 The number of regions is no longer a 
combinatorial invariant, but it depends in a subtle way on the coefficients.
Moreover, the regions are generally not contractible.

We present a practical software tool, called {\tt HypersurfaceRegions.jl} and implemented in the programming language {\tt Julia} \cite{BEKV}, 
whose input consists of
 $k$ polynomials in $n$ variables:
\begin{equation}
\label{eq:input} f_1,f_2,\ldots, f_k \,\in\, \RR[x_1,\ldots,x_n]. 
\end{equation}
The output is a list of all regions $C$ of the $n$-dimensional manifold
\begin{equation}
\label{eq:manifold}\mathcal{U} \,\, = \,\, \bigl\{ \, u \in \RR^n \,: \, f_1(u) \cdot f_2(u) \,\cdots\, f_k(u) \,\not= \,0 \,\bigr\}. 
\end{equation}
The list is grouped according to sign vectors $\sigma \in \{-1,+1\}^k$,
where $\sigma_i $ is the sign of $f_i$ on~$C$.
Unlike in the case of hyperplane arrangements,
each sign vector $\sigma$ typically
corresponds~to multiple regions.
For each region $C$ we find the Euler characteristic via a Morse function.

\begin{figure}
\begin{center}
\includegraphics[height = 5cm]{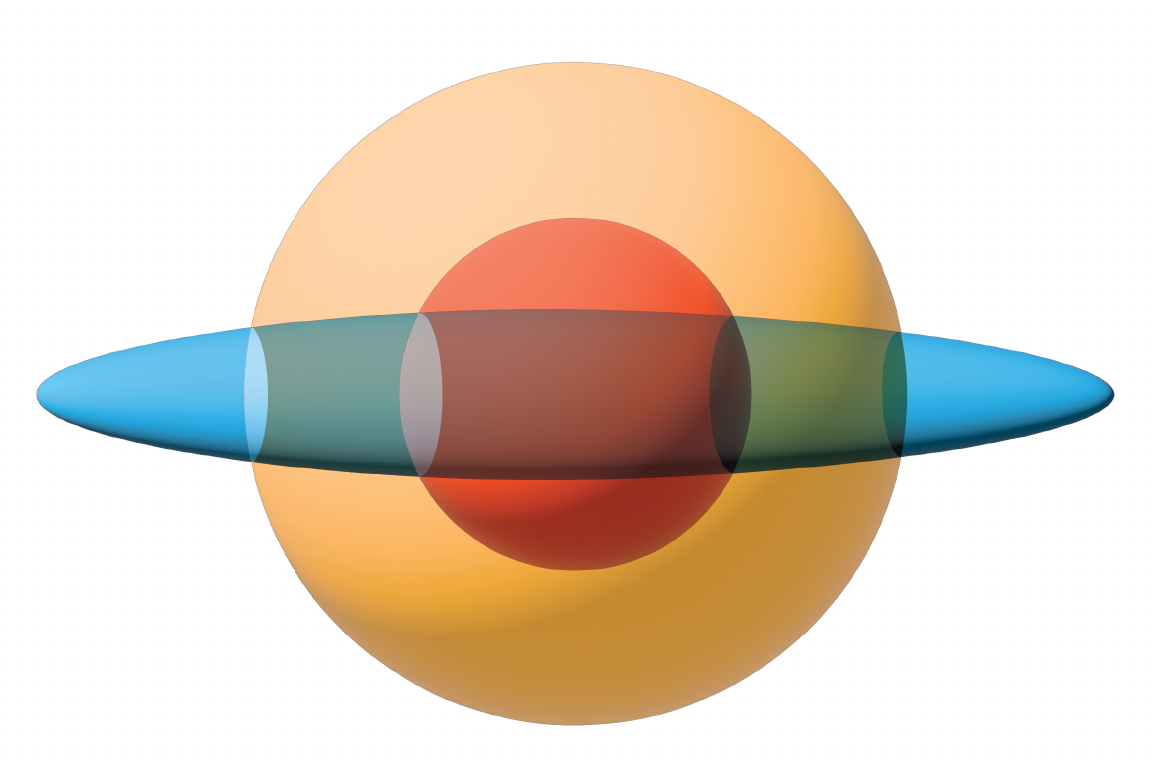}
\vspace{-0.1in}
\caption{\label{fig:concentric}
A skinny ellipse pierces two concentric spheres.
This arrangement has $8$ regions.}
\end{center}
\end{figure}

\begin{example}[$n\!=\!k\!=\!3$] \label{ex:skinny}
Figure \ref{fig:concentric} shows
two concentric spheres that are pierced by an ellipse:
$$ f_1 \,=\, x_1^2 + x_2^2 + x_3^2 - 1,\,\,
     f_2 \,=\, x_1^2 + x_2^2 + x_3^2 - 4,\,\,
     f_3 \, = \, 100 x_1^2 + 100 x_2^2 +  x_3^2 - 9. 
$$     
The threefold $\mathcal{U}$ has eight regions.
Five are contractible, with Euler characteristic $\chi = 1$.
The central region has  $\sigma = (-,-,-)$.
The sign vectors $(+,-,-)$ and $ (+,+,-)$ each contribute
two regions.
Two bounded regions are solid tori, with $\chi = 0$ and
$\sigma = (-,-,+),(+,-,+)$.
The unique unbounded region, with $\sigma = (+,+,+)$ and $\chi = 2$,
 is homotopic to the $2$-sphere.
\end{example}

Our software {\tt HypersurfaceRegions.jl} also features heuristics for deciding whether a region is bounded or unbounded, and which of the regions are fused when 
the hyperplane at infinity is added.
  In our Section \ref{sec4}, titled \emph{How to Use the Software}, we explain how this~works.
 
 \begin{example}[$n=2,k=21$]  \label{ex:cubicsurface}
 In del Pezzo geometry \cite[Section 3]{EGPSY}, it is known that 
 removing the  $27$ lines from a real cubic surface creates $130$ regions.
 We can confirm this by applying {\tt HypersurfaceRegions.jl} to the plane
 curves in the blow-up construction of the surface.
   Fix six general points in convex position in $\RR^2$.
 Let $f_1,\ldots,f_{15}$ be the lines
 spanned by pairs of  points and
 $f_{16},\ldots,f_{21}$ the ellipses spanned by five of the
 points. 
 This arrangement has
   $145$ regions, all contractible,
   of which $115$ are bounded and $30$ are unbounded. 
 Our software identifies these and fuses $15$ pairs of unbounded regions.
 It outputs the $130$ regions in $\RR \PP^2$.
 \end{example}

The method we present is an adaptation of the algorithm 
by Cummings, Hauenstein, Hong and Smyth \cite{CHHS}, which in turn is inspired by \cite{HRSS2020}.
Their setting is more general in that they allow semialgebraic sets
defined by both equations and inequalities.
We restrict ourselves to arrangement complements, defined by 
inequations $f_1 \not= 0,\ldots,f_k \not= 0$.
This enables us to offer tools in {\tt Julia} \cite{BEKV} that are easy to use
and widely applicable. Our implementation is based on
  the numerical algebraic geometry software {\tt HomotopyContinuation.jl} \cite{julia} and on the software {\tt DifferentialEquations.jl}~\cite{RN2017} for solving differential equations.

This paper is organized as follows.
In Section \ref{sec2} we introduce a Morse function for our problem. This is a 
modified version of the log-likelihood function associated with $f_1,\ldots,f_k$.
Every region of $\mathcal{U}$ 
contains at least one critical point.  We compute the critical points 
using  {\tt HomotopyContinuation.jl} .
These points
determine the Euler characteristic of each region.

 In Section \ref{sec3} we turn to the Mountain 
Pass Theorem \cite{Bis, MM}, which
 ensures that we obtain a connected graph in each region when tracking paths from
index $1$ critical points to index $0$ critical points. This tracking procedure is realized 
in our software by solving an~ODE using {\tt DifferentialEquations.jl}.
In short,  our software is built on flows in Morse theory.

Section \ref{sec4} explains how to run {\tt HypersurfaceRegions.jl}.
It is aimed at beginners with no
prior experience with numerical software. We demonstrate that {\tt HypersurfaceRegions.jl} can be used for
a wide range of scenarios from the mathematical spectrum.
In Section \ref{sec5} we report on test runs of
our software on generic instances, obtained by sampling
random hypersurfaces and random spectrahedra.
Our presentation emphasizes ease and simplicity.

\section{Log-likelihood as a Morse function}
\label{sec2}

The algorithm in \cite{CHHS} rests on Morse theory. We review some basics from the textbook \cite{Nic}.
Let $\mathcal{M}$ be an $n$-dimensional manifold.
A smooth function $g: \mathcal{M} \rightarrow \RR$ is a {\em Morse function} if the
 Hessian matrix $\bigl( \partial^2 g / \partial x_i  \partial x_j\bigr) $ is invertible at every critical point of $g$.
The number of positive eigenvalues of the Hessian matrix is 
the {\em index} of the critical point. Let $g$ be a Morse function which is
{\em exhaustive}, i.e.~the superlevel set $\{g \geq c\} = \{u \in \mathcal{M} : g(u) \geq c\}$
is compact for each $c \in \RR$.
If an interval $[a,b] \subset \RR$ contains no critical value of $g$
then the superlevel sets $\{g \geq c\}$ are diffeomorphic for  $c \in [a,b]$.
By contrast, suppose 
 $c$ is a critical value of $g$, corresponding to a unique critical point 
of index $\ell$. Then, for sufficiently small $\epsilon$,
the superlevel set $\{g \geq c- \epsilon\}$ is diffeomorphic
to $\{g \geq c+\epsilon \}$ with one $\ell$-handle of dimension $n$ attached.
Since $\ell$-handles can be shrunk to $\ell$-cells, one arrives at
the following result.

\begin{theorem}\label{thm:euler}
The manifold $\mathcal{M}$ is homotopy equivalent to a CW-complex with exactly one $\ell$-cell for every critical point of index $\ell$. In particular, the Euler characteristic of $\mathcal{M}$ equals
\begin{equation} 
\label{eq:chi} \chi(\mathcal{M}) \,=\,
\sum_{\ell = 0}^n (-1)^\ell \mu_\ell.
\end{equation}
where $\mu_\ell$ is the number of index $\ell$ critical points of the exhaustive Morse function
$g : \mathcal{M} \rightarrow \RR$.
\end{theorem} 

We shall apply Theorem \ref{thm:euler} to
the manifold $\mathcal{U}$ defined in (\ref{eq:manifold}).
This requires an exhaustive Morse function.
On each region of $\,\mathcal{U}$, this role will be played by the rational function
\begin{equation}
\label{eq:likelihood}
g(x) \,\, = \,\, \frac{ \prod_{i=1}^k |f_i(x)|^{s_i}} { q(x)^t } .
\end{equation}
For the denominator $q(x)$ we take a generic polynomial of degree $2$ that is strictly positive on $\RR^n$.
The exponents $s_1,\ldots,s_k,t$ are arbitrary positive integers which satisfy the inequality
\begin{equation}
\label{eq:stineq} \sum_{i=1}^k s_i \,{\rm degree}(f_i) \,\, < \,\, 2\, t. 
\end{equation}

\begin{proposition}
The function $g(x)$ is an exhaustive Morse function.
It is strictly positive on $\,\mathcal{U}$ and it is zero
on each of the hypersurfaces $\{f_i = 0\}$. It also vanishes at infinity in $\RR^n$.
\end{proposition}

The Morse function $g(x)$ in (\ref{eq:likelihood}) is called
 {\em master function} in the theory of arrangements \cite{CDFV} and
 {\em likelihood function} in algebraic statistics \cite{ST}.
Its logarithm is the {\em log-likelihood function}
\begin{equation}
\label{eq:loglikelihood}
{\rm log}(g(x)) \,\, = \,\, \sum_{i=1}^k s_i \cdot {\rm log} | f_i (x)| \, - \, t \cdot {\rm log}(q(x)). 
\end{equation}
The critical points of $g(x)$ are the solutions in $\CC^n$ of the
equation $\nabla\, {\rm log}(g(x)) = 0$.
The generic choice of the quadric $q(x)$ implies that all critical points
have distinct critical values.  The number of critical values is
the Euler characteristic of the very affine complex variety
\begin{equation}
\label{eq:veryaffine}
 \bigl\{\, u \in \CC^n \,:\, f_1(u) \cdot f_2(u) \,\cdots \,f_k(u) \cdot q(u) \,\not= \,0\, \bigr\}.
 \end{equation}
The construction of $g(x)$ ensures that each region of $\mathcal{U}$
contains at least one critical point. 
Our algorithm computes all real critical points, and it connects
them via the Mountain Pass Theorem. This will be explained in Section \ref{sec3}.
We now first
return to our three ellipsoids.

\begin{example}
Let $f_1,f_2,f_3$ be as in Example \ref{ex:skinny}, 
fix $s_1=s_2=s_3=1$ and $ t=4$, and define
$$ q \,\,=\,\, (x_1+2)^2 + (x_2-3)^2 + (x_3-3)^2 + (2x_1+x_2)^2 + 4. $$
This quadric is positive on $\RR^3$. Our Morse function for the complement of the three ellipses~is
$$ g \,\,=\,\, \frac{|f_1 f_2 f_3|}{q^4}. $$
The log-likelihood function ${\rm log}(g) = {\rm log}(f_1) + {\rm log}(f_2) + {\rm log}(f_3) - 4 \cdot{\rm log}(q)$
has $29$ complex critical points, so the threefold (\ref{eq:veryaffine})
has ML degree $29$. Among the critical points, $21$  are real. These serve as
 landmark points for the eight regions.
For instance, let $\mathcal{M}$  be the unbounded region.
It contains ten critical points, and their indices are
$0,0,1,1,1,1,2,2,2,2$.
In the notation of Theorem \ref{thm:euler}, we have
$\mu_0 = 2, \mu_1 = 4, \mu_2 = 4$, and hence
$\chi(\mathcal{M}) = \mu_0 - \mu_1 + \mu_2 = 2$.
\end{example}

The signed Euler characteristic of $\mathcal{U}$
is known as the maximum likelihood degree ({\em ML degree}) in
algebraic statistics \cite{CHKS, ST}. This is the number of
critical points we must compute. Explicitly, the logarithmic derivative of
$g(x)$ translates into the rational function equations
\begin{align*} 
\frac{s_1}{f_1}\frac{\partial f_1}{\partial x_1}\,+&\cdots+\frac{s_k}{f_k}\frac{\partial f_k}{\partial x_1}\,
+\,\frac{t}{q}\frac{\partial q}{\partial x_1}\,\,=\,\,0, \\
&\!\!\! \!\!\!\!\!\! \!\!\! \vdots \qquad \qquad \qquad \vdots \qquad \qquad \vdots  \\
\frac{s_1}{f_1}\frac{\partial f_1}{\partial x_n}\,+&\cdots+\frac{s_k}{f_k}\frac{\partial f_k}{\partial x_n}\,
+\,\frac{t}{q}\frac{\partial q}{\partial x_n}\,\,=\,\,0.
\end{align*}
The following a priori bound on the number of solutions was established in \cite[Theorem 1]{CHKS}.

\begin{proposition} \label{prop:catanese}
Let $d_1,\ldots,d_k$ be the degrees of the polynomials $f_1,\ldots,f_k$.
Then the ML degree of $\,\mathcal{U}$ is bounded above by the coefficient of $z^n$ in the 
rational generating function
\begin{equation}
\label{eq:cataneseseries}
\frac{(1-z)^n}{(1-z d_1)\,\cdots\, (1-z d_k)(1-2z)}.
\end{equation}
This bound is attained when the polynomials $f_i$ are generic relative to their degrees $d_i$.
\end{proposition}

\begin{example}
The ML degree in Proposition \ref{prop:catanese} equals
$\sum_{i=1}^k d_i^2+ \sum_{i<j} d_i d_j  +1$ for $n=2$,
and it equals $\sum_{i \leq j \leq l } d_i d_j d_l - \sum_{i \leq j} d_i d_j + 1$ for $n=3$.
See \cite[Section 4]{CHKS} for some nice geometry.
\end{example}

Since every region of our manifold $\mathcal{U}$ contains at least one critical point, we conclude:

\begin{corollary} \label{cor:bound}
The number of regions of $\,\mathcal{U}$ is at most the ML degree in 
Proposition~\ref{prop:catanese}.
\end{corollary}

The bound in Corollary \ref{cor:bound} is tight for linear polynomials.
This fact, and many more results on ML degrees, will be proved in a forthcoming article by
Berhard Reinke and Kexin Wang. We close this section by sketching a proof for 
hyperplanes in general~position.

Suppose $k \geq n$ and let $d_1 = \cdots = d_k = 1$. Then the generating function (\ref{eq:cataneseseries}) 
becomes
$$ \frac{1}{(1-z)^{k-n} \cdot (1-2z) }. $$
We find that the coefficient of $z^n$ in the Taylor expansion of this rational function equals
$$ 1 + k + \binom{k}{2} + \cdots + \binom{k}{n}. $$
This agrees with the number of all regions in a general arrangement of $k$ hyperplanes in $\RR^n$.
Thus, all complex critical points are real, and each critical point lies in its
own region.

\section{Mountain Passes and Path Tracking}
\label{sec3}

The Mountain Pass Theorem guarantees that all critical points in a region will be connected.
This theorem originates in the Calculus of Variations. We learned about its importance
for numerical computations in real algebraic geometry from the work of Cummings et al.~\cite{CHHS}.

We fix the Morse function $g$ on $\mathcal{U}$ that is given in (\ref{eq:likelihood}).
Since the quadric $q$ is generic, the function $g$ has
only finitely many critical points, and $g$ takes distinct values at these critical points.
Given any starting point  $x_0\in\mathcal{U}$, we will reach one of 
these critical points by numerical path tracking. This is done by
solving the ODE that describes the gradient flow
\begin{equation}
\label{eq:ode}
\dot{x}(t)\,=\,\nabla g(x)(t) ,\qquad
x(0)\,=\,x_0.
\end{equation}
If $x_0$ is chosen at random then the gradient flow will reach a critical point of index $0$.
Loosely speaking, hill climbing in the steepest direction will probably lead us to a mountain peak.

Let $p_1,p_2,\ldots,p_d \in \mathcal{U}$ denote the real critical points of $g$.
Our aim is to build a graph with vertex set $\{p_1,p_2,\ldots,p_d\}$ whose
connected components correspond to the regions of~$\mathcal{U}$.
For each critical point $p_i$, we compute the eigenvalues and eigenvectors of 
the Hessian  $H_{p_i}$. This is the symmetric $n \times n$ matrix of second
derivatives of $g(x)$ evaluated at $x = p_i$.

 Note that ${\rm det}(H_{p_i}) \not= 0$
because $g$ was constructed to be a Morse function.
An eigenvector $v$ of~$H_{p_i}$ is called
{\em unstable} if the corresponding eigenvalue is positive. If this holds then
$$ \quad g(p_i + \epsilon v) \,>\, g(p_i) \qquad \hbox{whenever $\,|\,\epsilon\,|\,$ is small.} $$
If $x_0 = p_i + \epsilon v$ is the starting point 
in (\ref{eq:ode}) then the ODE will lead us to
a critical point $p_j$ of~$g$ that is different from $p_i$.
Whenever this happens, our graph acquires the edge
$p_i \rightarrow p_j$.

We now explain why this method connects all critical points in a fixed region.
First of all, each $p_i$ is connected to some critical point of index $0$.
Suppose we start at a critical point~$p_{i_1}$ with positive index. There is a path from $p_{i_1}$ that limits to some critical point $p_{i_2}$ with $g(p_{i_2}) > g(p_{i_1})$. If $p_{i_2}$ has positive index, then the solution path from $p_{i_2}$ limits to some critical point $p_{i_3}$ with $g(p_{i_3}) > g(p_{i_2})$. Since there are only finitely many critical points, the process must terminate and the last critical point in the sequence $p_{i_1},p_{i_2},\ldots$ has index~$0$.

We are left to show that all index $0$ critical points in the same region will be connected.
To this end, we introduce one more tool from Morse theory.
The {\em stable manifold} of $p_i$ is 
$$
M(p_i)\,=\, \{p_i\}\,\cup\, \{ x\in \mathcal{U} \,\,|\,\, \nabla\, g(x)\neq 0 \text{ and the ODE solution starting from } x \text{ limits to }p_i \}.
$$
The dimension of $M(p_i)$ is the number of stable eigenvectors of $H_{p_i}$.
The stable manifolds for critical points of index $0$ have full dimension $n$ in $\mathcal{U}$.
So, for each region $C$ of $\mathcal{U}$, we~have
\begin{equation}
\label{eq:bigcup}
C\,\,= \! \bigcup_{p_i \in C\text{ index 0}} \! \overline{M(p_i)},
\end{equation}
where the closure is taken in $\mathcal{U}$.
Consider any two critical points $p_{i_\alpha}$ and $p_{i_\beta}$ of index~$0$ in $C$.
Since $C$ is connected,   (\ref{eq:bigcup}) implies that there is a sequence of 
index $0$ critical points $p_{i_1},\ldots,p_{i_s}$ such that
the intersections $\overline{M(p_{i_\alpha})}\cap \overline{M(p_{i_1})}$ 
and $ \overline{M(p_{i_s})} \cap \overline{M(p_{i_\beta})}$ are non-empty,
and also $\, \overline{M(p_{i_j})}\cap \overline{M(p_{i_{j+1}})} \,$ is non-empty
 for~$j=1,\ldots,s-1$.
Corollary \ref{cor:connect} below tells us that each of these intersections
contains at least one critical point of index~$1$.
Thus any pair of index $0$ critical points in $C$ is connected through 
index $1$ critical points. Therefore,
the connectivity of the finite graph we are building for $C$
will be ensured by Corollary~\ref{cor:connect}.

The Morse function $g: C \rightarrow \RR$ satisfies the {\em Palais-Smale condition}.
This means that  every sequence
$\{x_j\}$ in $C$ which satisfies $\,\lim_{j \rightarrow \infty} g(x_j)=\alpha \in \RR\,$
and $\, \lim_{j \rightarrow \infty} \nabla\, g(x_j)=0 \,$ has a convergent subsequence.
The limit of such a convergent subsequence is a critical point in $C$ with critical value $\alpha$.
The Palais-Smale condition holds in our situation because $g$ is positive on $C$, it tends to zero
on the boundary, and $g^{-1}([\alpha_1,\alpha_2])$ is compact for $0 < \alpha_1 < \alpha_2$.
 
 A {\em path} between two points $a$ and $b$ in $C$ is a continuous function $\gamma: [0,1]\mapsto C$ with $\gamma(0)=a$
and $\gamma(1)=b$. We write $\Gamma_{a,b}$ for the set of all such paths.
A set $S \subset C$ {\em separates} $a$ and $b$ if every path  $\gamma \in \Gamma_{a,b}$ contains some point in~$S$. We write $g(\gamma) = \{ g(\gamma(t)) : t \in [0,1] \}$.
The Palais-Smale condition for $g$ is a hypothesis for the next result, which is
   \cite[Theorem~3]{MM}.

\begin{theorem}\label{thm:pass}(Mountain Pass Theorem)
Fix a closed subset $S \subset C$ which separates $a$ and~$b$
and satisfies $\,{\rm max}\bigl( g(x): x \in S\bigr) < {\rm min}\bigl( g(a),\,g(b)\bigr)$.
Then  $\,\omega \,=\, {\rm sup} \bigl( {\rm min}(g(\gamma)) : \gamma \in \Gamma_{a,b})\,$
is a critical value of $g$, and its preimage $g^{-1}(\omega)$ contains a unique critical point of index $1$.
 \end{theorem}

\begin{corollary} \label{cor:connect}
Let $p_i$ and $p_j$ be index $\,0$ critical points of the Morse function $g$ on $C$, and 
suppose that $\,S=\overline{M(p_i)}\cap \overline{M(p_j)}\,$ is non-empty.
Then $\,S$ contains a critical point of index~$1$.
\end{corollary}

\begin{proof}
The set $S$ satisfies the hypotheses of Theorem \ref{thm:pass} for
$a = p_i$ and $b = p_j$. Let $p_\ell$ be the critical point which is
promised by Theorem  \ref{thm:pass}. The point $p_\ell$ necessarily lies in $S$.
\end{proof}

We have shown that the regions $C$ of $\mathcal{U}$ can be identified by
gradient ascent from index~$1$ critical points. Here we start in the two
 directions determined by the unique unstable eigenvector.
The resulting graph on the index $0$ critical points is guaranteed to be connected.
We now summarize our approach. Algorithm \ref{alg:Hauenstein}
 forms the basis of {\tt HypersurfaceRegions.jl}.

\begin{algorithm}[h]
\caption{Computing the regions of an arrangement of $k$ hypersurfaces in~$\RR^n$} \label{alg:Hauenstein}
\KwIn{Polynomials $f_1,f_2,\ldots,f_k $ in $\RR[x_1,\ldots,x_n]$.}

\KwOut{All sign patterns in $\{-,+\}^k$ that are realized by $f_1,\ldots,f_k$, and each region with that sign pattern,
along with its  Euler characteristic.}

Calculate the critical points $\{p_1,\ldots,p_d\}$ of $\log g(x)$ and record the sign vectors 
$$\sigma = \bigl( {\rm sign}(f_1(p_i)),\ldots,{\rm sign}(f_k(p_i))\bigr)\, \in\, \{-1,+1\}^k \,\,\,\,
 {\rm for} \,\,\, i=1,\ldots,d. $$

Calculate the Hessian matrix $H_{p_i}$ of each critical point $p_i$, verify that it is invertible, and compute its index and the unstable eigenvectors.

\ForEach{sign pattern $\sigma$}{
 Identify the set $\{p_{i_1},\ldots,p_{i_j}\}$ of all critical points with sign pattern~$\sigma$.
 
Initialize a graph $G_\sigma$ which has this set as its vertex set.

\uIf{
there is only one vertex $p_{i_r}$ of index $0$ in $G_\sigma$}{Add edges between all index $>0$ critical points and $p_{i_r}$ to $G_\sigma$.
}

\uElse{
For each index $1$ critical point $p_{i_\ell}$,
identify the unstable eigenvector $v$, and solve the ODE (\ref{eq:ode}) starting from $p_{i_\ell}+\epsilon v$ and $p_{i_\ell}-\epsilon v$ for small $\epsilon>0$.

Compute the limit points (one or two). Add edge(s) from $p_{i_\ell}$ to these in $G_{\sigma}$.

For each $p_{i_\ell}$ of index $> 1$, pick an unstable eigenvector $v$, solve (\ref{eq:ode})  from 
$p_{i_\ell}\!+\!\epsilon v$ for small $\epsilon>0$, and add to $G_\sigma$ the edge from $p_{i_\ell}$ to the limit~point.
}

Identify the connected components of the graph $G_\sigma$. Compute the Euler characteristic of the corresponding regions using the formula in \eqref{eq:chi}.
}
\end{algorithm}

\begin{proof}[Correctness proof of Algorithm \ref{alg:Hauenstein}]
Each region in $\mathcal{U}$ has constant sign pattern for the evaluation of $(f_1,\ldots,f_k),$ so it suffices to track the ODE paths between critical points with the same sign pattern. 
Tracking the ODE along the positive and negative unstable eigenvector directions of index $1$ critical points connects all index $0$ critical points in the same region. 

A sign pattern $\sigma$ is realizable if and only if it is attained by some critical point.
We now fix $\sigma$, and we consider all critical points with sign pattern $\sigma$.
For every index $>1$ critical point $p_{j_1}$, the ODE path connects it with some critical point $p_{j_2}$ with $g(p_{j_2})>g(p_{j_1})$. If again $p_{j_2}$ has index $>1$, then the ODE path from $p_{j_2}$ limits to some critical point $p_{j_3}$ with $g(p_{j_3}) > g(p_{j_2})$. Since there are only finitely many critical points, the process 
terminates, and $p_{j_1}$ is connected to some index $0$ or $1$ critical point via a sequence of paths.
Therefore, the connected components of the graph $G_\sigma$ correspond exactly to the regions in $\mathcal{U}$
with sign pattern $\sigma$. The statement about the Euler characteristic 
follows from Theorem~\ref{thm:euler}.
\end{proof}

Frequently one is interested in the regions of an arrangement in 
projective space $\RR \PP^n$.
Algorithm \ref{alg:projective} addresses this point.
To begin with, we take the  same input as before, namely $f_1,f_2,\ldots,f_k
\in \RR[x_1,\ldots,x_n]$. We further assume that
Algorithm \ref{alg:Hauenstein} has already been run.

\begin{algorithm}[ht]
\caption{regions of a projective arrangement of $k$ hypersurfaces in $\RR \PP^n$} \label{alg:projective}
\KwIn{All data that were used and computed in Algorithm \ref{alg:Hauenstein}.}

\KwOut{A list of all regions in projective space $\RR \PP^{n}$, grouped by realizable
sign pattern pairs in $\{-,+\}^{k}$.}

Homogenize $f_1,\ldots,f_k  \in \RR[x_1,\ldots,x_n]$ to $f_1'(x_0,x_1,\ldots,x_n),\ldots,f_k'(x_0,x_1,\ldots,x_n)$.

Compute the regions of $\,\mathcal{U}_{\infty}=\bigl\{\,u \in \RR^{n-1}:  \, f_1'(0,1,u)
 \,\cdots \,f_k'(0,1,u)\neq 0 \bigr\}$ and record the list $\{ q_1,\ldots,q_{d'}\}$ of representatives from each region.

\ForEach{$q_i$}{
    Find the largest absolute value among all zeros of the polynomials
    $f_j(t,t q_i)$ for $i=1,\ldots,d$ and $j=1,\ldots,k$. Fix a real number $\lambda$
    larger than that value.

    Solve the ODE (\ref{eq:ode}) with starting points $\lambda(1,q_i)$ and 
    $-\lambda(1,q_i)$. Record the limiting critical points. The regions of these points  are connected in~$\RR \PP^n$.}

The regions not visited in the previous step are either {\em bounded} or {\em undecided}.
\end{algorithm}

\begin{proof}[Correctness proof of Algorithm \ref{alg:projective}]
All unbounded regions of $\mathcal{U}$ intersect the hyperplane  at infinity, denoted $\{x_0=0\}$.
If two of them are connected in $\RR \PP^n$, then their intersections with
$\{x_0=0\}$ share a region of $\mathcal{U}_{\infty}$.
We sample points $q_i$ from each region of $\mathcal{U}_{\infty}$.
The choice of $\lambda$ guarantees that
 $\lambda(1,q_i)$ and $-\lambda(1,q_i)$ lie in
 unbounded regions of $\RR^n$ which are adjacent to the
region of $q_i$ at infinity. In particular, 
$\lambda(1,q_i)$ and $-\lambda(1,q_i)$ lie in the same projective region,
and this region contains the unbounded affine regions that are identified
by the ODE~(\ref{eq:ode}). One unbounded
region can get matched multiple times in this process.
Hence, even three or more regions in $\RR^n$
can get fused to the same region in $\RR \PP^n$.
\end{proof}

In our implementation of Algorithm \ref{alg:projective} we discard critical points that are close to the hypersurface $\mathbb R^{n-1}\backslash \mathcal U_\infty$. 
These points cause numerical problems when solving the corresponding ODE. Hence, that
part of our software is based on a heuristic.

We now explain the distinction between {\em bounded}
and {\em undecided}, alluded to in Step 7. It is a feature of 
numerical computations that tangencies and singularities are hard to detect.

\begin{example}[$n=2,k=1$]
Fix a real number $\epsilon$ and consider the plane curves given~by
$$ f_\epsilon(x,y) \,=\, y^2 - (x-1)(x^2+\epsilon)
\qquad  {\rm and} \qquad
g_\epsilon(x,y) \,=\, y^2+\epsilon x^2 - x . $$
 Each curve is viewed as an arrangement in $\RR^2$ with $k=1$.
 The arrangements $\{f_\epsilon\}$ and $\{g_\epsilon\}$ have two regions for $\epsilon \geq 0$
and three regions for $\epsilon < 0$. 
Note the topology changes when we pass from $\epsilon > 0$ to $\epsilon = 0$.
The cubic $f_0$ acquires an isolated point, and the conic $g_0$ becomes tangent to the
line at infinity. Both regions of $\{g_0\}$ are unbounded, but only one of them
is recognized as unbounded by Algorithm \ref{alg:projective}. The
interior of the parabola remains undecided.
\end{example}

We identify bounded regions as follows.
We fix a small parameter $\delta > 0$.
A region in Step~7 is declared  {\em bounded} if it is
contained in $ -1/\delta < x_1 < 1/\delta$.
In practice, we compute critical points of our Morse function on the hyperplanes $x_1 = 1/\delta$ and $x_1 = -1/\delta$. 
We process these critical points similar to what is done in Algorithm \ref{alg:projective}. 
Namely, for each critical point $q$ on the hyperplane $x_1 = 1/\delta$, we find the largest value among all zeros of the polynomials
$f_j(t,t q)$ for $j=1,\ldots,k$ that are smaller than $1/\delta$. 
Fix a real number~$\lambda$
larger than that value, but smaller than $1/\delta$. 
We then solve the ODE (\ref{eq:ode}) with starting points~$\lambda \cdot (1,q)$ and record the limiting critical points.  We proceed similarly for the other hyperplane $x_1 = -1/\delta$.

If a region in $\RR^n$  is not visited during this process, then it does not intersect either hyperplane.
If its critical points satisfy $ -1/\delta < x_1 < 1/\delta$, then it is bounded. 
If a region is not bounded or unbounded, then it is declared 
{\em undecided}.  Thus, undecided regions are
close to being tangent to the hyperplane at infinity, where ``close'' refers to
the tolerance~$\delta$.

\section{How to Use the Software}
\label{sec4}

Our software is easy to use, even for beginners.
We here offer a step-by-step introduction. For a complete overview over our implementation we refer to the
online documentation at
$$ \hbox{
 \url{https://juliaalgebra.github.io/HypersurfaceRegions.jl/}.}
$$
The first step is to download the programming language {\tt Julia}.
 One navigates to the website  \url{https://julialang.org/downloads/} for the latest version.
After it is downloaded and installed, we start {\tt Julia} and enter the package manager by pressing the {\tt ]} button. 
Once in the package manager, type {\tt add HypersurfaceRegions} and hit enter. This installs
{\tt HypersurfaceRegions.jl}. One leaves the package manager by pressing the back space button. 
To load our software into the current Julia session, type the following command:
\begin{lstlisting}[language=code]
  using HypersurfaceRegions
\end{lstlisting}
Now, we are ready to use the implementation. Let's compute our first arrangement!!

\begin{figure}[h]
\begin{center}
\includegraphics[height = 7.7cm]{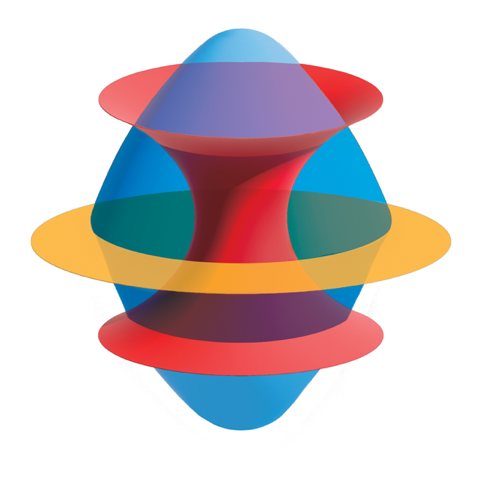}
\vspace{-0.24in}
\caption{\label{fig:hyperboloid}
A hyperboloid and two paraboloids.
The paraboloids intersect in a circle that spans a plane.
This arrangement of four surfaces has $12$ regions in $\RR^3$. Four are contractible.
}
\end{center}
\end{figure}

Being ambitious, we 
start with an example in $\RR^3$.
Consider the four surfaces in Figure~\ref{fig:hyperboloid}.
The yellow plane is defined by $z=0$. The three quadrics are
invariant under rotation around the $z$-axis. The two paraboloids are
defined by $z = x^2+y^2-1$ and $z = -x^2-y^2+1$.
The hyperboloid is defined by $x^2+y^2-\frac{3}{2} z^2 = \frac{1}{4}$.
We input this arrangement into our software:

\begin{lstlisting}[language=code]
@var x y z
f = [z, x^2 + y^2 - 1 + z, x^2 + y^2 - 1 - z, x^2 + y^2 - 1/4 - 3/2z^2]
regions(f)
\end{lstlisting}

The output  for this input is shown on the left in Figure \ref{fig:hypout}.
 We see that $\mathcal U$ has $12$ regions,
 and each of them is uniquely identified by its sign pattern
 $\sigma \in \{-,+\}^4$.   Six regions lie outside the hyperboloid,
 which means that $\sigma_4 = +$. Each of them is
 homotopy equivalent to a circle, so $\chi = 0$.
 The other six regions lie inside the hyperboloid, which means
 $\sigma_4 = -$.

Two of the regions inside the hyperboloid are unbounded and also circular ($\chi = 0$).
Finally, there are two bounded regions and two regions tangent to the hyperplane at infinity.
  They are all contractible ($\chi = 1$). 
We compute this information by setting a flag as follows:
\begin{lstlisting}[language=code]
  regions(f; bounded_check = true)
\end{lstlisting}
The output of
{\tt HypersurfaceRegions.jl} for this input is shown on the right in Figure \ref{fig:hypout}.
This also reminds us that the number of real critical points 
depends on $q(x)$, which is random.

\begin{figure}[h]
\begin{center}
\includegraphics[height = 10.3cm]{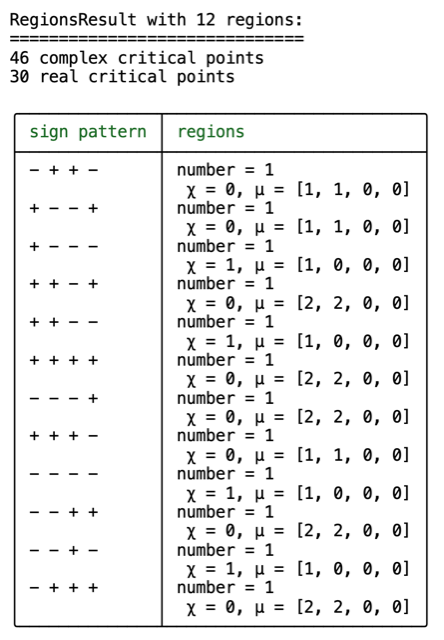} \quad \,
\includegraphics[height = 9.8cm]{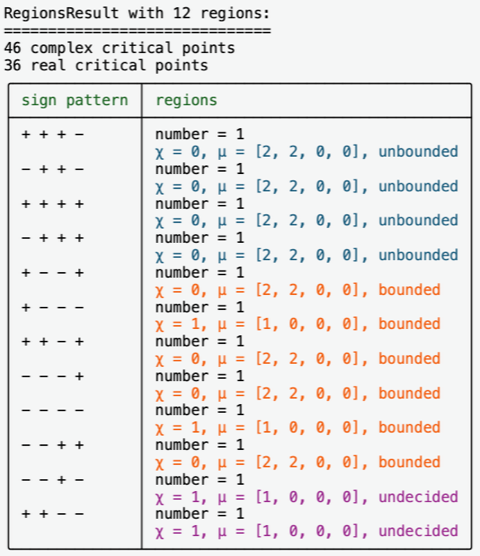}
\caption{\label{fig:hypout}
The output of {\tt HypersurfaceRegions.jl}
for the arrangement shown in Figure \ref{fig:hyperboloid}.}
\end{center}
\end{figure}

One application of our method is the study of
discriminantal arrangments. Such arrangements
arise whenever one is interested in the topology
of real varieties that depend on parameters.
The one-dimensional case of this is known as
{\em Real Root Classification}, where one expresses
the number of real roots of a zero-dimensional
system  as a function of parameters.

To use the current version of {\tt HypersurfaceRegions.jl} in this context,
it is assumed that the discriminant is known
and that it breaks up into smaller factors. Here is an example.

\begin{example}[Polynomial of degree $8$] \label{ex:dis8}
Consider the following polynomial in one variable~$t$:
$$ f(t) \,\,=\,\,  t^8\,+\,x \cdot t^7\,+\,(x+y)\cdot t^6\,+\,t^5\,+\,(x+y) \cdot t^4\,+\,(y+1) \cdot t^3
\,+\,t^2\,+\,(x+y)\cdot t. $$
The coefficients depend linearly on two parameters $x$ and $y$. The discriminant of
$f(t)$ is a polynomial of degree $14$ in $x$ and $y$. This discriminant factors into
four irreducible factors. Two of these have multiplicity two.
We input the four factors into our software as follows.
\begin{lstlisting}[language=code]
@var x y
f = [x + y, 
     23x^6 + 60x^5*y + 50x^4*y^2 + 16x^3*y^3 + 3x^2*y^4 - 78x^5 
        - 336x^4*y - 478x^3*y^2 - 284x^2*y^3 - 76x*y^4 - 12y^5 - 87x^4 
        - 144x^3*y + 54x^2*y^2 + 180x*y^3 + 68y^4 + 28x^3 + 24x^2*y 
        - 58x*y^2 - 56y^3 - 87x^2 - 300x*y - 208y^2 - 78x - 72y + 23,
    x + 3y + 1,
    5x^2 + 4x*y + y^2 - 6x - 4y + 5]
regions(f)
\end{lstlisting}

\begin{figure}[h]
\begin{center}
\includegraphics[height = 9cm]{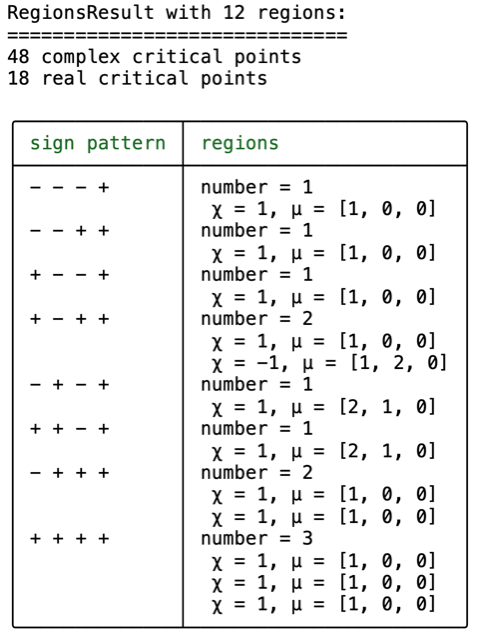} 
\caption{\label{fig:uni8out}
Output  for the discriminantal arrangement in Example \ref{ex:dis8}.}
\end{center}
\end{figure}

The output is shown in  Figure \ref{fig:uni8out}.
Eight of the $16$ sign patterns are realizable.
The parameter plane is divided into $12$ regions.
For instance, $\sigma = (+,-,+,+)$ contributes two regions,
one contractible and one with two holes ($\chi=-1$).
The latter parametrizes
polynomials with only two real roots.
One sample point in that region is $(x,y) = (-1,5)$.
\end{example}

You are now invited to run the following code and to match its output
  with Figure~\ref{fig:concentric}:
   \begin{lstlisting}[language=code]
@var x y z 
f = [x^2 + y^2 + z^2 - 1, x^2 + y^2 + z^2 - 4, 100*x^2 + 100*y^2 + z^2 - 9]
R = regions(f, bounded_check = true)
\end{lstlisting}
A range of additional features is available in {\tt HypersurfaceRegions.jl}.
For instance,
one finds the critical points in the $i$-th region of the output by running the following command:
\begin{lstlisting}[language=code]
Ri = regions(R)[i]
critical_points(Ri)
\end{lstlisting}
One can also locate the region to which a given point (e.g. the origin) belongs, as follows:
\begin{lstlisting}[language=code]
p = [0, 0, 0]
membership(R, p)
\end{lstlisting}

We next discuss a few implementation details.  We use \emph{monodromy} \cite{DHJLLS}, implemented in {\tt HomotopyContinuation.jl} \cite{julia}, for computing the critical points of $g(x)$. The ODE solver from~{\tt DifferentialEquations.jl}~\cite{RN2017} is used for
the Morse flows~(\ref{eq:ode})
between critical points.

To find critical points, we solve the  rational function equations~$\nabla\, {\rm log}(g(x)) = 0$ in
\eqref{eq:likelihood}.
 These are $n$ equations in $n$ variables $x_1,\ldots,x_n$ and $k+1$ parameters $u=(s_1,\ldots,s_k,t)$. For the monodromy, we need a start pair $(x_0, u_0)$ such that $x_0$ is a solution for $\nabla\, {\rm log}(g(x)) = 0$ 
  with parameters $u_0$. Since $\nabla\, {\rm log}(g(x))$ is linear in $u$, we can use linear algebra to get such a start pair. Indeed, we can sample a random point $x_0$ and compute $u_0$ by solving linear equations. However, this works only if $k+1>n$, and if $\nabla\, {\rm log}(g(x_0))$ has full rank. To avoid these issues, we use the following trick. If $k+1\leq n$, we define new auxiliary parameters $v=(v_1,\ldots,v_{n-k})$. Using monodromy, we solve the system $\nabla\, {\rm log}(g(x)) - (Au + Bv)=0$, where $A\in\mathbb C^{n\times (k+1)}$ and $B\in\mathbb C^{n\times (n-k)}$ are random matrices. Afterwards, we track the homotopy $\nabla\, {\rm log}(g(x)) - \lambda\cdot(Au + Bv)=0$ from $\lambda = 1$ to $\lambda = 0$ using the solutions we have just computed. The \emph{parameter continuation theorem} \cite{Sommese:Wampler:2005,  BOROVIK2025102373} asserts that this approach works. 

We modified the default options in {\tt HomotopyContinuation.jl} for the monodromy loops.
In {\tt HypersurfaceRegions.jl}, a computation finishes
if the last $10$ new monodromy loops have not provided any new solutions. The 
number $10$ is a conservative parameter here. Its purpose is 
to increase the probability of finding all critical points. 
The tolerance parameter~$\delta$ for bounded regions,
 as described in the end of Section \ref{sec3},
is set to the default value~$\delta = 10^{-5}$.

\section{Experiments}
\label{sec5}

This final section is dedicated to systematic experiments with our software.
We experimented by running
{\tt HypersurfaceRegions.jl} on random instances, drawn from two classes:
\begin{enumerate}
\item arrangements defined by random polynomials; \vspace{-0.2cm}
\item arrangements defined by random spectrahedra. 
\end{enumerate}
The first class is self-explanatory. To create
(\ref{eq:input}),  we choose $f_i$ at random
from some probability distribution on the space of
inhomogeneous
polynomials of degree $d$ in $n$ variables.

The second class requires an explanation.
We fix positive integers $n$ and $m$,
and we draw $n+1$ samples 
$A_0,A_1,\ldots,A_n$ from the
 $\binom{m+1}{2}$-dimensional
space of symmetric $m \times m$ matrices. 
We consider the arrangement
defined by the $k = 2^m-1$ principal minors of the matrix
\begin{equation}
\label{eq:A(x)} A(x) \,\, = \,\, A_0 \, +\, x_1 A_1 \, + \,\cdots \, + \, x_n A_n . 
\end{equation}
The distinguished region where all $k$ principal minors are positive 
is the {\em spectrahedron}.

The stratification of symmetric matrices
by signs of principal minors was studied by
Boege, Selover and Zubkov in \cite{BSZ}. 
We explore  random low-dimensional slices of the
{\em Boege-Selover-Zubkov stratification}.
How many regions get created, and what is their topology?

\begin{example}[$n \! = \! m \! = \!3$]
\label{ex:cayley}
A prominent spectrahedron is the {\em elliptope}, which is given~by
$$ A(x,y,z) \,\, = \,\,
 \begin{bmatrix} 
  1 & x & y \\
  x & 1 & z\\
  y & z & 1\end{bmatrix}.
  $$
  The $2 \times 2$ minors of this matrix give the six facet equations of
  the $3$-dimensional cube $[-1,1]^3$.
  Our arrangement lives in $\RR^3$, and it  consists
of these six planes plus the cubic surface 
\begin{equation}
\label{eq:cayleycubic} {\rm det}(A) \,\, = \,\,
2 xyz - x^2 - y^2 - z^2 + 1.
\end{equation}
For an illustration see Figure \ref{fig:elliptope}.
Only four of the six facets are shown 
for better visibility.

\begin{figure}[h]
\begin{center}
\includegraphics[height = 6.7cm]{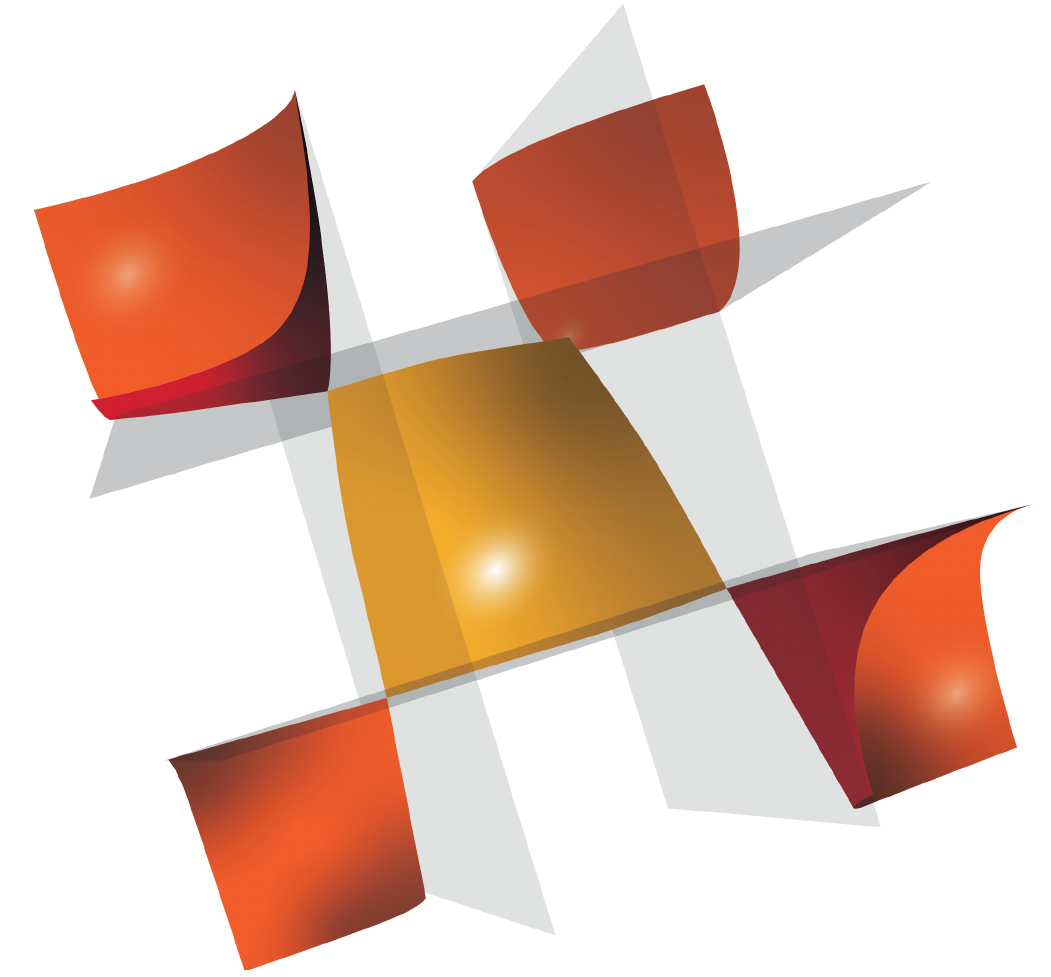}
\caption{\label{fig:elliptope}
Arrangement given by the elliptope
and the facet planes of the surrounding cube.}
\end{center}
\end{figure}

We run {\tt HypersurfaceRegions.jl} on this $k=7$ instance.
The arrangement has $43$ regions. All have Euler characteristic $1$.
The six facet planes of the cube divide $\RR^3$
into $27 = 1 \!+\! 8 \!+\! 12 \!+\! 6$ regions. One is bounded
(the cube itself) and $26$ are unbounded cones. 
The surface (\ref{eq:cayleycubic}) divides the cube into $5$ regions,
and it divides four of the cones into $4$ unbounded regions,
for a total of $27 \!+\! 4 \!+\! 12 = 43$ regions.
The projective arrangement has $39$ regions.
\end{example}

We now describe our experiments with arrangements of generic hypersurfaces.
We run {\tt HypersurfaceRegions.jl} on $k$ polynomials in $n$ variables
of degree $d=2$ and $d=3$, where the
coefficients are drawn independently from the standard Gaussian.
Each experiment
is carried out $N$ times, where $N$ is
inverse proportional to the running time for each instance.

Our findings are presented in Table \ref{tab:zwei} for $d=2$ and in
Table \ref{tab:drei} for $d=3$.
For each row we fix the parameters $n$ and $k$.
The time is the average running time per instance.
This is measured in seconds. 
The fourth column concerns the total number of
regions.
We report the minimum number and the maximum number
across all instances. 
The fifth column similarly reports the range 
for the number of realizable sign vectors $\sigma$.
The sixth column gives the maximal number of
regions per fixed sign pattern.
And, finally, in the last column we report the
minimum and maximum observed for the
Euler characteristics $\chi$ of the regions.

\begin{table}[h]
\begin{center}
\begin{tabular}{|c|c|c|c|c|c|c|}
\hline
$n$ & $k$ &  time & min-max \#reg & min-max $\# \sigma$ & max $\#$reg$/\sigma$ & min-max $\chi$ \\ 
\hline
$3$ & $8$ & $15.5$ & $46, 325$ & $30 , 196$ & $11$ & $-5, 2$ \\
\hline
$3$ & $9$ &  $28.9$ & $72, 462$ & $41, 319$ & $7$ & $-2, 2$ \\
\hline
$3$ & $10$ &  $68.1$ & $108, 612$ & $83 , 437$ & $8$ & $-3 , 1$ \\
\hline
$4$ & $4$ &  $15.4$ & $9, 38$ & $8, 16$ & $7$ & $-6, 6$ \\
\hline
$4$ & $5$ &  $32.6$ & $19, 87$ & $15, 32$ & $7$ & $-7 , 3$ \\
\hline
$4$ & $6$ &  $72.3$ & $47, 192$ & $32, 64$ & $9$ & $-7 , 2$ \\
\hline
$5$ & $5$ & $142.9$  &  $34, 67$& $24, 32$  & $8$  & $-8, 4$\\
\hline
\end{tabular}
\vspace{-0.12in}
\end{center}
\caption{\label{tab:zwei}
Random affine arrangements defined by $k$ quadrics in $\RR^n$.}
\end{table}

Let us discuss the first row in Table \ref{tab:zwei}.
The code runs $15.5$ seconds, and it identifies $\rho $ distinct regions in $\RR^3$,
where $46 \leq \rho \leq 325$. The number of realizable sign
patterns is at most $196$.
This is less than the number  $2^k = 2^8 = 256$ of all sign patterns.
There were up to $11$ regions per sign pattern. The
Euler characteristic of any region was between $-5$ and $2$.

\begin{table}[h]
\begin{center}
\begin{tabular}{|c|c|c|c|c|c|c|}
\hline
$n$ & $k$ &  time & min-max \#reg & min-max $\# \sigma$ & max $\#$reg$/\sigma$ & min-max $\chi$ \\ 
\hline
$3$ & $5$ &  $25.0$ & $34, 135$ & $27, 32$ & $8$ & $-2, 2$ \\
\hline
$3$ & $6$ &  $59.5$ & $67,186$ & $52,64$ & $7$ & $-2, 2$ \\
\hline
$3$ & $7$ &  $107.5$ & $123, 280$ & $87, 126$ & $7$ & $-2, 2$ \\
\hline
$4$ & $4$ &  $127.4$ & $19, 47$ & $16, 16$ & $6$ & $-4, 3$ \\
\hline
$4$ & $5$ &  $348.5$ & $46, 119$ & $32, 32$ & $8$ & $-4, 3$ \\
\hline
$5$ & $3$ &  $290.9$ & $8, 11$ & $8, 8$ & $2$ & $-7, 4$ \\
\hline
\end{tabular}
\vspace{-0.12in}
\end{center}
\caption{\label{tab:drei}
Random affine arrangements defined by $k$ cubics in $\RR^n$.}
\end{table}

Table \ref{tab:drei} presents analogous results for cubics
in $\RR^n$ where $n=3,4,5$.
The second-last row concerns quintuples of cubics in $\RR^4$.
It takes less than six minutes to identify all regions,
of which there are up to $119$. We observed a narrow range
$\{-4,\ldots,3\}$ of Euler characteristics.

Generic instances have no
tangencies to the hyperplane at infinity. To experiment
with that issue, we can try $k$ paraboloids in $\RR^n$.
Each paraboloid contributes one undecided region touching the hyperplane at infinity.
Here is an instance of $k\!=\!4$ paraboloids for  $n\!=\!3$:

\begin{lstlisting}[language=code]
@var x y z
f = [3 + x + 3*y - z + (1 + 2*x + 4*y - 4*z)^2 + (2 + 3*x + 2*y + 3*z)^2,
     3 + x + 3*z + (3 - 3*x - 2*z)^2 + (3 + 3*x + 3*y + 4*z)^2,
     2 - 2*x - 2*y - 3*z + (2 - x + 4*z)^2 +  (2 + 3*x + y + 2*z)^2,
     1 - 3*x + 3*y - 3*z + (1 - 2*y + 2*z)^2 + (2 + x + 4*y)^2 ]
regions(f, bounded_check = true)
\end{lstlisting}
The output in Figure \ref{fig:paraout} consists of six regions, each with a unique sign pattern. Notice that the label ``undecided'' does \emph{not} mean we claim this region touches the hyperplane at infinity. It means that our algorithm could not decide whether this region is bounded or unbounded. 

\begin{figure}[h] 
\begin{center}
    \includegraphics[height=5.6cm]{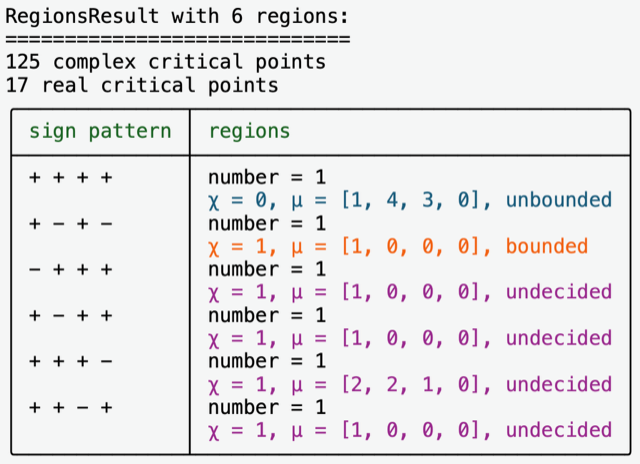}
    \caption{Output of {\tt HypersurfaceRegions.jl} for four paraboloids in $\RR^3$.}
    \label{fig:paraout}
\end{center}
\end{figure}

We now turn to random spectrahedra.
We sampled symmetric $m \times m$ matrices $A_0,\ldots,A_n$
where the entries are independent standard Gaussians.
The   principal minors of the matrix $A(x)$ define
an arrangement of $k = 2^m-1$ hypersurfaces in the affine space $\RR^n$.
We ran {\tt HypersurfaceRegions.jl} on this input. 
Our results are summarized in Table \ref{tab:eins}.
For each row we fix the parameters $n$ and $m$.
The columns have the same meaning as before.
For instance, the fifth column  reports the range 
for the number of realizable sign vectors $\sigma$.

\begin{table}[h]
\begin{center}
\begin{tabular}{|c|c|c|c|c|c|c|}
\hline
$n$ & $m$ &  time & min-max \#reg & min-max $\# \sigma$ & max $\#$reg$/\sigma$ & min-max $\chi$ \\ 
\hline
$2$ & $2$ &  $3.3$ & $4, 8$ & $4, 6$ & $3$ & $\,\,\,1, 1$ \\
\hline
$2$ & $3$ & $4.4$ & $38, 58$ & $24, 35$ & $9 $& $-1, 1$ \\
\hline
$3$ & $3$ &  $13.4$ & $80, 122$ & $34, 38$ & $21$ & $-2, 1$ \\
\hline
$4$ & $3$ & $46.3$ & $117, 150$ & $38, 38$ & $26$ & $-2, 1$ \\
\hline
\end{tabular}
\vspace{-0.12in}
\end{center}
\caption{\label{tab:eins}
Random arrangements in $\RR^n$ defined by $n+1$ symmetric $m \times m$ matrices.}
\end{table}

We also performed computations for $m \geq 4$,
but we do not report them because we ran into
 numerical issues. Frequently, the 
Hession of  (\ref{eq:loglikelihood})
  at some critical point is almost singular.

\bigskip
\bigskip

\footnotesize{
{\bf Authors' addresses:}

\bigskip

Paul Breiding, Universit\"at Osnabr\"uck 
 \hfill {\tt \href{mailto:pbreiding@uni-osnabrueck.de}{pbreiding@uni-osnabrueck.de}}

Bernd Sturmfels, MPI MiS Leipzig  \hfill {\tt \href{mailto:bernd@mis.mpg.de}{bernd@mis.mpg.de}}

Kexin (Ada) Wang, Harvard University \hfill {\tt \href{mailto:kexin_wang@g.harvard.edu}{kexin$\_$wang@g.harvard.edu}}

\end{document}